\documentclass[11pt]{article}
\usepackage[utf8]{inputenc}

\usepackage{fullpage}

\usepackage{enumerate}

\usepackage{amsmath}
\usepackage{amssymb}
\usepackage{amsthm}
\usepackage{mathtools}
\usepackage{physics}

\usepackage{bbm}
\usepackage[numbers]{natbib}

\usepackage{hyperref}
\hypersetup{
    colorlinks=true,
    linkcolor=blue,
    citecolor=blue,
    urlcolor=blue,
}

\theoremstyle{plain}
\newtheorem{theorem}{Theorem}[section]
\newtheorem{proposition}{Proposition}[section]

\newtheorem{lemma}[theorem]{Lemma}
\newtheorem{corollary}[theorem]{Corollary}

\theoremstyle{definition}
\newtheorem{definition}[theorem]{Definition}

\theoremstyle{remark}

\newcommand   {\E}          {\operatorname*{\mathbb{E}}}
\renewcommand {\P}          {\operatorname*{\mathbb{P}}}
\renewcommand {\max}        {\operatorname*{max}}

\newcommand   {\argmax}     {\operatorname*{argmax}}

\newcommand   {\intersect}  {\,\cap\,}

\newcommand   {\supp}       {\operatorname{supp}}
\newcommand   {\paren}  [1] {\left( #1 \right)}
\newcommand   {\set}    [1] {\left\{ #1 \right\}}

\begin{document}
\title{Delegated Stochastic Probing}
\author{
    Curtis Bechtel\thanks{supported by NSF Grants CCF-1350900 and CCF-2009060.}\\
    Department of Computer Science \\
    University of Southern California \\
    {\tt bechtel@usc.edu}
      \and
    Shaddin Dughmi\thanks{supported by NSF CAREER Award CCF-1350900 and NSF Grant CCF-2009060.} \\
    Department of Computer Science \\
    University of Southern California \\
    {\tt shaddin@usc.edu}
}
\date{December, 2020}
\maketitle

\begin{abstract}
    Delegation covers a broad class of problems in which a principal doesn't have the resources or expertise necessary to complete a task by themselves, so they delegate the task to an agent whose interests may not be aligned with their own. Stochastic probing describes problems in which we are tasked with maximizing expected utility by ``probing'' known distributions for acceptable solutions subject to certain constraints. In this work, we combine the concepts of delegation and stochastic probing into a single mechanism design framework which we term delegated stochastic probing. We study how much a principal loses by delegating a stochastic probing problem, compared to their utility in the non-delegated solution. Our model and results are heavily inspired by the work of Kleinberg and Kleinberg in ``Delegated Search Approximates Efficient Search.'' Building on their work, we show that there exists a connection between delegated stochastic probing and generalized prophet inequalities, which provides us with constant-factor deterministic mechanisms for a large class of delegated stochastic probing problems. We also explore randomized mechanisms in a simple delegated probing setting, and show that they outperform deterministic mechanisms in some instances but not in the worst case.
\end{abstract}

\section{Introduction}

The division of labor and responsibility, based on expertise, is a defining characteristic of efficient organizations and productive economies. In the context of economic decision-making, such division often manifests through \emph{delegation} scenarios of the following form: A decision maker (the \emph{principal}), facing a multivariate decision beset by constraints and uncertainties, tasks an expert (the \emph{agent}) with collecting data, exploring the space of feasible  decisions, and proposing a solution.

As a running example, consider the leadership of a firm delegating some or all of its hiring decisions to an outside recruitment agency. When the principal and the agent have misaligned utilities --- such as when the agency must balance the firm's preferences with its own preferences over, or obligations towards, potential hires --- the principal faces a mechanism design problem termed \emph{optimal delegation} (see e.g. \cite{holmstrom78, armstrong10}). When the underlying optimization problem involves multiple inter-dependent decisions, such as when hiring a team which must collectively cover a particular set of skills, and when data collection is constrained by logistical or budget considerations,  the problem being delegated fits in the framework of \emph{stochastic probing}, broadly construed (see e.g.~\cite{singla_thesis}).

The present paper is concerned with the above-described marriage of optimal delegation and stochastic probing. We restrict attention to protocols without payments,  drawing our inspiration from the recent work of Kleinberg and Kleinberg \cite{KK18}. The underlying (non-delegated) problem faced by the principal in their ``distributional model'' is the following: facing $n$ i.i.d rewards, select the ex-post best draw. As for their ``binary model'', there are $n$ random rewards with binary support, and a cost associated with sampling each; the goal is to adaptively sample the rewards and select one, with the goal of maximizing the ex-post selected reward less sampling costs. For both models, they show that delegating the problem results in a loss of at most half the principal's utility. Their analysis in both cases is through a reduction to the (classical) single-choice \emph{prophet inequality} problem, and in particular to the threshold stopping rule of Samuel-Cahn~\cite{SC84}.

Both the distributional and binary models of \cite{KK18} can be  viewed as stochastic probing problems, the former being trivial in the absence of delegation, and the latter corresponding to a special case of the well-studied box problem of Weitzman~\cite{weitzman79}. A number of stochastic probing problems have been known to reduce to \emph{contention resolution schemes} (e.g. \cite{probing, probing_gns16,matching_heaven,ROCRS,OCRS}), which in turn reduce to generalizations of the prophet inequality \cite{OCRS_prophet}. This suggests that the results of \cite{KK18} might apply more broadly.

It is this suggestive thread which we pull on in this paper, unraveling what is indeed a broader phenomenon. We study optimal delegation for a fairly general class of stochastic probing problems with combinatorial constraints, and obtain delegation mechanisms which approximate, up to a constant, the principal's non-delegated utility. Building on recent progress in the literature on stochastic optimization, our results reduce delegated stochastic probing to \emph{generalized prophet inequalities} of a particular ``greedy'' form, as well as to the notion of \emph{adaptivity gap} (e.g. \cite{adaptivity, multivalue_adaptivity}).

\subsection{Our Model}

Our model features a collection of \emph{elements}, each of which is associated with a (random) \emph{utility} for each of the principal and the agent. We assume that different elements are independently distributed, though the principal's and the agent's utilities for the same element may be correlated. We allow constraining both the sampled and the selected set of elements via \emph{outer} and \emph{inner} constraints, respectively. Each constraint is a downwards-closed set system on the ground set of elements. A \emph{probing algorithm} for an instance of our model adaptively \emph{probes} some set of elements subject to the outer constraint, learning their associated utilities in the process. The algorithm then selects as its \emph{solution} a subset of the probed elements satisfying the inner constraint. We assume that, for both the principal and the agent, utility for a solution is the sum of its per-element utilities.

To situate the non-game-theoretic component of our model within the literature on stochastic probing problems, note that we allow an arbitrary utility distribution for each element, rather than a binary-support distribution characterizing ``feasibility''.  Moreover, unlike ``probe and commit'' models, we also allow our algorithm to select its solution after all probes are complete.  In both these respects, our model is akin to the stochastic multi-value probing model of \cite{multivalue_adaptivity}. As for our game-theoretic modeling, we assume that the utility distributions, as well as the inner and outer constraints, are common knowledge. The realized utilities, however, are only observed by the agent upon probing.

In the traditional (non-delegation) setting, the principal implements the probing algorithm optimizing her own utility, in expectation. In the delegation setting, the principal and agent engage in the following Stackelberg game. The principal moves first by \emph{committing} to a \emph{policy}, or \emph{mechanism}. Such a policy is a (possibly randomized) map from a set of \emph{signals} to solutions satisfying the inner constraint, with each element in the solution labeled with its (presumptive) utility for both the principal and the agent. Moving second, the agent probes some set of elements subject to the outer constraint, and maps the observed utilities to a signal. The outcome of the game is then the solution which results from applying the principal's policy to the agent's signal. We assume that the principal and agent utilities are additive across elements in the solution, so long as it is labeled with the true per-element utilities. Otherwise, we assume that the principal can detect this discrepancy and effectively ``quit'' the game, imposing a utility of zero for both parties. We adopt the perspective of the principal, who seeks a policy maximizing her expected utility. The agent, naturally, responds with a strategy maximizing his own expected utility given the policy.

By an argument analogous to that in \cite{KK18}, which we prove in our general setting for completeness' sake, we can restrict attention to \emph{single-proposal mechanisms}. In a deterministic single-proposal mechanism, the set of signals is a ``menu'' of \emph{acceptable}  (labeled) solutions satisfying the inner constraint, as well as a ``null'' signal which in our setting we can take to be the empty set. The agent, facing such a mechanism, without loss simply implements a probing algorithm to compute a ``proposed'' solution, tagging each element in the solution with its observed utilities, and ensuring that the solution is acceptable to the principal. We also consider randomized single-proposal mechanisms, where the menu consists of acceptable \emph{lotteries} (i.e., distributions) over (labeled) solutions, and an agent's probing algorithm proposes a lottery on the menu.

\subsection{Our Results}

We study delegation mechanisms which approximate the principal's non-delegated utility. We refer to the best multiplicative approximation factor as the \emph{delegation gap} of the associated instance.

Our main set of results concern the design of deterministic single-proposal mechanisms which prove constant delegation gaps for natural classes of inner and outer constraints. Our approach is modular, and reduces a (constructive) $\alpha \beta$ bound  on  the delegation gap to a (constructive) $\alpha$ generalized prophet inequality of a particular form on the inner constraint, and a (constructive) bound of $\beta$ on the adaptivity gap associated with the outer constraint and the rank function of the inner constraint. Drawing on recent work in \cite{OCRS}, which derives prophet inequalities of our required form, and in \cite{adaptivity, multivalue_adaptivity}, which bounds the adaptivity gap, we obtain constant bounds on the delegation gap for instances of our model with a variety of inner and outer constraints such as matroids and their intersections, among others.

We also begin an exploration of randomized single-proposal mechanisms, where the principal's menu consists of acceptable lotteries over solutions. We show that, even in the simple setting of no outer constraint and a $1$-uniform inner constraint, there are instances for which randomized mechanisms significantly outperform optimal deterministic ones. Nevertheless, there exist worst-case instances where both deterministic and randomized mechanisms suffer a $1/2$ delegation gap. We leave open whether randomized mechanisms can lead to better bounds on the worst-case delegation gap for more intricate classes of inner and outer constraints.

\subsection{Additional Discussion of Related Work}

Since the economic literature on delegation is extensive, we only describe a select sample here. The groundwork for the formal study of optimal delegation in economics was initially laid by Holstrom~\cite{holmstrom78, holmstrom80}. Subsequent work in economics has considered a variety of optimization problems as the task being delegated (e.g. \cite{alonso08,melumad91,armstrong10}). We mention the work of Kovac and Mylovanov~\cite{lottery_delegation} as being related to our results in Section~\ref{lottery-mechanisms}: To our knowledge, they were the first to examine the power of randomized mechanisms for delegation.

Most relevant to the present paper is the aforementioned work of Kleinberg and Kleinberg~\cite{KK18}, who examine approximations for optimal delegation. Their distributional model is closely related to the model of Armstrong and Vickers~\cite{armstrong10}, and the optimization problem being delegated in their binary model is a special case of Weitzman's box problem~\cite{weitzman79}. Both optimization problems fit nicely in the general literature on stochastic probing (see e.g. \cite{singla_thesis}), motivating our examination of delegated stochastic probing more broadly.

Also related is the recent work of Khodabakhsh et al~\cite{threshold_delegation}, who consider a very general model of delegation with discrete actions and states of the world, and an agent who fully observes the state (no outer constraints or sampling costs). They show optimal delegation to be NP-hard and examine limited ``bias'' assumptions under which simple threshold mechanisms are approximately optimal. Notably, they don't impose sampling constraints on the agent and their approximations are with respect to the optimal delegation policy rather than the optimal non-delegated policy. For these reasons, our results are not directly comparable.

The optimization problems being delegated in our model fit in the broad class of \emph{stochastic probing} problems. We do not attempt a thorough accounting of this literature, and instead refer the reader to related work discussions in \cite{singla_thesis, multivalue_adaptivity}. To our knowledge, the term ``stochastic probing'' was originally coined by Gupta and Nagarajan \cite{probing}, though their binary probe-and-commit model is quite different from ours. More closely related to us are the models of \cite{multivalue_adaptivity, adaptivity}, which capture stochastic probing problems with multi-valued reward distributions, no commitment, and combinatorial inner and outer constraints.

As previously mentioned, our work draws on the literature on prophet inequalities. The foundational result in this setting is the (single-choice) prophet inequality of Krengel, Sucheston, and Garling~\cite{KS77, KS78}. \emph{Generalized prophet inequalities}, with respect to various combinatorial constraints, adversaries, and arrival models, have received much attention in the last decade (e.g. \cite{HKS07, matroid_prophet, prophet_easy, OCRS}); the associated body of work is large, and we recommend the survey by \cite{prophet_survey_lucier}. Closely related to generalized prophet inequalities are \emph{contention resolution schemes} (see e.g. \cite{CRS, OCRS, ROCRS}), with reductions going in both directions~\cite{OCRS_prophet}. Key to our results are the ``greedy'' generalized prophet inequalities, derived through ``greedy'' contention resolution, by Feldman et al~\cite{OCRS}.

Finally, we briefly elaborate on the relationship between our model and the two models of Kleinberg and Kleinberg \cite{KK18}. The natural variant of their binary model which replaces sampling costs with combinatorial constraints on the set of samples (outer constraints, in our nomenclature) fits squarely in our model. Their distributional model, which allows $n$ i.i.d. samples from a distribution over utility pairs, initially appears to be a special case of ours. However, our principal is afforded additional power through their ability to distinguish elements by name alone. Nevertheless, we recover their main result as a special case of ours by observing that our mechanism treats elements symmetrically.

\section{Preliminaries}
\label{preliminaries}

Sections \ref{set-systems-and-matroids}, \ref{prophet-inequalities}, and \ref{adaptivity-gap} include brief introductions to some of the key ideas and notations used in this paper. Notably, Section \ref{prophet-inequalities} defines the key notion of ``greedy'' prophet inequalities.

\subsection{Set Systems}
\label{set-systems-and-matroids}

A \emph{set system} is a pair $(E, \mathcal{I})$ where $E$ is a finite set of \emph{elements} and $\mathcal{I} \subseteq 2^{E}$ is a family of \emph{feasible} sets. We focus on \emph{downwards-closed} set systems, satisfying the following two conditions: (1) $\emptyset \in \mathcal{I}$, i.e. the empty set is feasible, and (2) if $T \in \mathcal{I}$ then $S \in \mathcal{I}$ for all $S \subseteq T$, i.e. any subset of a feasible set is feasible. Matroids, matching constraints, and knapsack constraints are all examples of downwards-closed set systems.

For a set system $\mathcal{M}=(E,\mathcal{I})$ and $F \subseteq E$, we use $\mathcal{M}|F = (F,\mathcal{I} \intersect 2^F)$ to denote the \emph{restriction} of $\mathcal{M}$ to $F$.

\subsection{Prophet Inequalities}
\label{prophet-inequalities}

A \emph{generalized prophet inequality problem} is given by a set system $\mathcal{M} = (E, \mathcal{I})$, and for each element $e \in E$ an independent random variable $X_e$ supported on the nonnegative real numbers. Here we adopt the perspective of a \emph{gambler}, who is given $\mathcal{M}$ and the distributions of the random variables $\set{X_e}_{e \in E}$ in advance, then  encounters the elements $E$ in an order chosen by an \emph{adversary}. On encountering $e$, the gambler observes the realization $x_e$ of the random variable $X_e$, and must immediately decide whether to \emph{accept} $e$, subject to the accepted set $S$ of elements remaining feasible in $\mathcal{M}$.  The gambler seeks to maximize his utility $x(S)=\sum_{e \in S} x_e$, and in particular to compete with a \emph{prophet} who knows the realization of all random variables in advance. If the gambler can guarantee an $\alpha$ fraction of the prophet's utility in expectation, we say that we obtain a generalized prophet inequality with a factor of $\alpha$.

For each possible realization $x_e$ of $X_e$, we refer to the pair $(e, x_e) \in E \times \mathbb{R}_+$ as an \emph{outcome}. When the gambler accepts $e \in E$ given a realization $x_e$ of $X_e$, we also say the gambler accepts the outcome~$(e,x_e)$.

Although it is most common to consider an adversary who fixes an order of the elements upfront, some recent work has investigated much more powerful adversaries \cite{matroid_prophet, OCRS}. In this paper, we are interested in the \emph{almighty adversary}, who knows in advance the realizations of all random variables as well as any random coin flips used by the gambler's strategy. The almighty adversary can perfectly predict the future and choose a truly worst-case ordering.

Key to our results is the notion of a ``greedy'' strategy for the gambler. We take inspiration from \cite{OCRS}, who defined greedy online contention resolution schemes, and extend their definition to prophet inequality problems.

\begin{definition}
    Fix any instance of a generalized prophet inequality problem. A \emph{greedy} strategy for the gambler is described by a downwards-closed family $\mathcal{A} \subseteq 2^{E \times \mathbb{R}_+}$ of sets of outcomes.  A gambler following greedy strategy $\mathcal{A}$ accepts an outcome $(e,x_e)$ if and only if the set of all accepted outcomes remains in $\mathcal{A}$.
\end{definition}

We note that Samuel-Cahn's~\cite{SC84} threshold rule for the single-choice prophet inequality is greedy, and its competitive factor of $\frac{1}{2}$ holds for the almighty adversary \cite{roughgarden_twenty_lectures}. More generally, Feldman et al. \cite{OCRS} show that there exist constant-factor greedy prophet inequalities against the almighty adversary for many classes of constraints.

\subsection{Adaptivity Gap}
\label{adaptivity-gap}

Another key notion we will use is the adaptivity gap for stochastic set function optimization problems. For a detailed introduction, see \cite{adaptivity}.

We consider maximizing a stochastic set function $f: 2^E \to \mathbb{R}_+$ constrained by a downwards-closed set system $\mathcal{M} = (E, \mathcal{I})$. We assume $f$ is determined by a collection $\set{X_e}_{e \in E}$ of independent random variables, with the stipulation that $f(S)$ does not depend on any random variables $X_e$ for which $e \notin S$.\footnote{In other words, one can evaluate $f(S)$ given access to the realizations of the random variables $\set{X_e}_{e \in S}$.} We are tasked with ``probing'' some $S \subseteq E$, feasible for $\mathcal{M}$, with the goal of maximizing $f(S)$. An \emph{adaptive} algorithm for this problem probes elements one at a time, where probing $e$ results in learning the realization of $X_e$. Such an algorithm can use the realizations of probed variables to decide on a next element to probe. A \emph{non-adaptive} algorithm chooses the set $S$ all at once, independently of the random variables $\set{X_e}_{e \in E}$.
The \emph{adaptivity gap} is the minimum (worst-case) ratio of the expected value of the optimal non-adaptive algorithm versus the expected value of the optimal adaptive algorithm.

In \cite{adaptivity}, Asadpour and Nazerzadeh showed that the adaptivity gap for instances with monotone submodular functions and matroid constraints is $1-\frac{1}{e}$. Furthermore, they provided an efficient non-adaptive algorithm that achieves this bound. Finally, in \cite{multivalue_adaptivity}, Bradac et al. showed that the adaptivity gap is constant for instances with ``prefix-closed'' constraints (which include all downward-closed constraints) and functions that are the weighted rank function of the intersection of a constant number of matroids.

\section{Model}

\subsection{Formal Definition}
\label{formal-definition}

\begin{definition}
    \label{def:deterministic-model}
    
    An instance $I$ of the \emph{delegated stochastic probing problem} consists of: two players, which we will call the \emph{principal} and the \emph{agent}; a ground set of elements $E$; mutually independent distributions $\mu_e$ with support in $\mathbb{R}_+ \times \mathbb{R}_+$ for each element $e \in E$; an \emph{outer} constraint $\mathcal{M}_{\text{out}} = (E, \mathcal{I}_{\text{out}})$ with feasible sets $\mathcal{I}_{\text{out}}$; and an \emph{inner} constraint $\mathcal{M}_{\text{in}} = (E, \mathcal{I}_{\text{in}})$ with feasible sets $\mathcal{I}_{\text{in}}$.
\end{definition}

Given such an instance, we will additionally define: $(X_e, Y_e) \sim \mu_e$ as random variables denoting the utilities for the principal and agent of element $e$; $\Omega$ as the set of \emph{outcomes} $(e, x, y)$ for all $e \in E$ and all $(x, y) \in \supp(\mu_e)$; and $\Omega_{\text{in}} \subseteq 2^\Omega$ as the family of all sets of outcomes whose elements are distinct and feasible in the inner constraint. For convenience, we will also overload notation by considering $x$ and $y$ to be utility functions for the principal and agent. Given any subset of outcomes $T \subseteq \Omega$, let $x(T) = \sum_{(e, x, y) \in T} x$ and $y(T) = \sum_{(e, x, y) \in T} y$ be the total utility of outcomes in $T$. Similarly for any subset of elements $F \subseteq E$, let $x(F) = \sum_{e \in F} X_e$ and $y(F) = \sum_{e \in F} Y_e$ be random variables representing the randomized total utility of elements in $F$.

A natural mechanism that the principal might choose to implement is called a single-proposal mechanism. Here, the principal describes the space of solutions she is willing to accept, and then the agent uses this information to search the solution space and propose a single feasible solution.

In the deterministic single-proposal setting, the principal first commits to a family of sets of outcomes $\mathcal{R} \subseteq \Omega_{\text{in}}$ and announces $\mathcal{R}$ to the agent. The sets in $\mathcal{R}$ are called \emph{acceptable}, and the principal's choice of $\mathcal{R}$ is called their \emph{policy} (or \emph{mechanism}). After learning $\mathcal{R}$, the agent will select elements to probe, so long as each element is probed at most once and the set of probed elements is feasible in $\mathcal{M}_{\text{out}}$. We allow the agent to probe adaptively, deciding what to do next based on previously probed elements. Let's say that they probe elements $F \subseteq E$ and obtain outcomes $S \subseteq \Omega$. The agent will then choose some set of outcomes $T \subseteq \Omega$ and \emph{propose} it to the principal. If $T$ is acceptable and also a subset of $S$ then the principal and agent receive $x(T)$ and $y(T)$ utility, respectively. Otherwise, they both receive $0$ utility.

In the above-described mechanism design setting, we assume that both the principal and agent act to maximize their expected utility. We also assume that all parameters of the problem, except for the realizations of the random variables, are common knowledge.

We note that, similar to the setup in \cite{KK18}, our model assumes that our agent cannot benefit from lying, say by labeling an element $e$ with utilities other than $X_e$ and $Y_e$, or by proposing an element he has not probed. We argue that this is a natural assumption to make: In many applications we foresee (e.g., a firm hiring an employee, or exploring some mergers), a proposal will be accompanied by an easy to verify proof of the claimed utilities (e.g., in the form of a CV for the applicant, or a detailed analysis of the merger).

As in \cite{KK18}, we compare delegation mechanisms against the optimal \emph{non-delegated} strategy. By non-delegated strategy, we mean the strategy of the principal when they act as both the principal and agent (i.e. they have power to probe and propose as well as accept outcomes).

Given any $F \subseteq E$, let $u(F)$ be the optimal utility of the non-delegating principal when they probe elements in $F$ and accept their own favorite set of outcomes, and let $v_{\mathcal{R}}(F)$ be the utility of the delegating principal with policy $\mathcal{R}$ when the agent probes elements in $F$ and proposes their favorite acceptable set of outcomes. We can write $u$ and $v_{\mathcal{R}}$ as
\begin{align*}
    u(F) &= \max_{G \subseteq F, G \in \mathcal{I}_{\text{in}}} x(G) \\
    v_{\mathcal{R}}(F) &= x\paren{\argmax_{G \subseteq F, \Omega_G \in \mathcal{R}} y(G)},
\end{align*}
where $\Omega_G \subseteq \Omega$ is the set of outcomes from the probed set of elements $G$. In the case of ties in the definition of $v_{\mathcal{R}}$, our results hold for arbitrary (even adversarial) tie-breaking.

\begin{definition}
    Fix any instance of delegated stochastic probing. Let $F^*$ be a random variable containing the elements probed by an optimal adaptive non-delegating principal, and let $F^*_{\mathcal{R}}$ be a random variable containing the elements probed by an optimal adaptive agent under policy $\mathcal{R}$. Then for any policy $\mathcal{R}$ and $\alpha \in [0, 1]$, we say that $\mathcal{R}$ is an \emph{$\alpha$-policy} for this instance if
    \begin{equation*}
        \E v_{\mathcal{R}}(F^*_{\mathcal{R}}) \ge \alpha \E u(F^*).
    \end{equation*}
\end{definition}

\begin{definition}
    The \emph{delegation gap} of a family of instances of delegated stochastic probing is the minimum, over all instances in the family, of the maximum $\alpha$ such that there exists an $\alpha$-policy for that instance. This gap measures the fraction of the principal's non-delegated utility they can achieve when delegating.
\end{definition}

\subsection{Signaling Mechanisms}
\label{single-proposal-mechanisms}

Having formally defined the model, we will now describe a broad generalization of single-proposal mechanisms, called \emph{signaling mechanisms}, and show that these mechanisms don't provide the principal with any additional power. Note that this discussion is inspired by Section 2.2 from \cite{KK18}, and we simply extend their work to our model.

A signaling mechanism allows the principal to ask the agent for more (or different) information than just a proposed solution. The principal will then take this information and transform it into a solution, which they will accept. One might suspect that expanding the space of mechanisms in this way would give the principal more power. However, as we will show, this isn't the case even for a broad class of delegation models, which we will now define formally.

\begin{definition}
    \label{def:generalized-delegation-problem}
    An instance of the \emph{generalized delegation problem} consists of two players called the \emph{principal} and the \emph{agent}, a state space $S$, a solution space $\Psi$, a set $\mathcal{P}$ of \emph{probing strategies} for the agent, a \emph{signaling function} $\sigma$ which maps $\mathcal{P} \times S$ to strings, a utility function $x : S \times \mathcal{P} \times \Psi \to \mathbb{R}_+$ for the principal, and a utility function $y : S \times \mathcal{P} \times \Psi \to \mathbb{R}_+$ for the agent. We require that there is a \emph{null solution} $\bot \in \Psi$ such that $x_{s,p}(\bot) = y_{s,p}(\bot) = 0$ for all $s \in S$ and $p\in \mathcal{P}$.
    
    We assume the state of the world is some $s \in S$ a-priori unknown to the principal and the agent, though they may have prior information. The agent obtains information about $s$ by applying a probing strategy $p \in \mathcal{P}$ to obtain a signal $\sigma_p(s)$.  For a state $s \in S$, a probing strategy $p \in \mathcal{P}$ chosen by the agent, and a solution $\psi \in \Psi$, we associate a utility of $x_{s,p}(\psi)$ and $y_{s,p}(\psi)$ for the principal and the agent, respectively.
\end{definition}

We note that the above definition generalizes the delegation problems of Definition \ref{def:deterministic-model}. In particular: the state space $S$ represents all possible realizations of per-element utilities of the principal and the agent; the solution space $\Psi$ is the family of feasible subsets of outcomes $\Omega_{\text{in}}$, where $\bot$ is the empty set of outcomes;  $\mathcal{P}$ corresponds to probing algorithms which respect the outer constraint; $\sigma_p(s)$ is the set of outcomes obtained by invoking algorithm $p$ in state $s$; both utility functions depend on the state $s \in S$ and the probing algorithm  $p \in \mathcal{P}$, evaluating to $0$ for solutions $\psi$ that are inconsistent with the state $s$, or if the probing algorithm $p$ applied to $s$ does not the probe the elements in $\psi$.

Given a generalized delegation problem, we define signaling mechanisms as follows.

\begin{definition}
    \label{def:signaling-mechanism}
    Fix some instance of the generalized delegation problem. A \emph{signaling mechanism} proceeds in the following manner. The principal starts by choosing some signal space $\Sigma$ of strings and a solution function $\psi : \Sigma \to \Psi$, and the agent responds by choosing a probing strategy $p \in \mathcal{P}$ and a \emph{reporting function} $\tau$ from strings to $\Sigma$. Once these choices have been made, the agent will probe the underlying state $s$ to obtain a signal $\sigma=\sigma_p(s)$, then transform this into a new signal $\tau= \tau(\sigma)$ which he reports to the principal. The principal maps the reported signal to a solution $\psi(\tau)$, which they will accept.
\end{definition}

Notice that this model can be made to capture the design of randomized delegation mechanisms by extending $\Psi$ to the space $\Delta(\Psi)$ of distributions (henceforth \emph{lotteries}) over solutions, and extending both utility functions to lotteries by taking expectations.

We contrast this broad definition of signaling mechanisms with the comparatively simple single-proposal mechanisms.

\begin{definition}
    \label{def:single-proposal-mechanism}
    Fix an instance of the generalized delegation problem. A \emph{single-proposal mechanism} is a special case of signaling mechanism in which the principal chooses some set $\mathcal{R} \subseteq \Psi$ of acceptable outcomes, then sets $\Sigma = \Psi$ and $\psi(R) = R$ if $R \in \mathcal{R}$ and $\psi(R) = \bot$ otherwise.
\end{definition}

Intuitively, in a single proposal mechanism the principal declares a menu of acceptable solutions. The agent then proposes a solution, which is accepted if it is on the menu, and replaced with the null solution otherwise. Now we will show that single-proposal mechanisms are just as powerful as signaling mechanisms. In particular, for every signaling mechanism there is a single-proposal mechanism which selects the same solution and the same probing strategy for each state of nature, at equilibrium. This lemma is a simple extension of \cite[Lemma 1]{KK18} to the our generalized delegation model.

\begin{lemma}
    \label{lem:single-proposal-mechanisms}
    Fix an instance of the generalized delegation problem, as well as the agent's prior distribution $\mu$ on states $S$. For any signaling mechanism $M = (\Sigma, \psi)$ and a corresponding best response strategy $(p, \tau)$ for the agent, there exists a single-proposal mechanism $M' = (\Sigma', \psi')$ and a corresponding best response $(p, \tau')$ such that $(\psi \,\circ\, \tau \,\circ\, \sigma_p) (s) = (\psi' \,\circ\, \tau' \,\circ\, \sigma_p) (s)$ for all states $s \in S$.
\end{lemma}
\begin{proof}
    Take any signaling mechanism $M = (\Sigma, \psi)$ with best response $(p, \tau)$ by the agent. Let $\mathcal{R} = \psi(\Sigma)$ be the set of all possible outputs from this mechanism and let $M' = (\Sigma', \psi')$ be the single-proposal mechanism defined by $\mathcal{R}$, i.e. $\Sigma' = \Psi$ and $\psi'$ is such that $\psi'(R) = R$ if $R \in \mathcal{R}$ and $\psi'(R) = \bot$ otherwise. Finally, let $\tau'= \psi \circ \tau$.
    
    Notice that the range of $\tau'$ is contained in $\psi(\Sigma)=\mathcal{R}$, so by definition of $\psi'$ and $\tau'$ it follows that $\psi \circ \tau = \psi' \circ \tau'$. Therefore, it is also the case that $(\psi \circ \tau \circ \sigma_p)(s) = (\psi' \circ \tau' \circ \sigma_p)(s) $ for all $s \in S$. Now we must show that $(p, \tau')$ is a best-response strategy to mechanism $M'$. Consider any valid alternative strategy $(p^*, \tau^*)$. We aim to show that
    \begin{equation}
        \label{eq:single-proposal-ineq}
        \E_s y_{s,p^*}(\psi' \circ \tau^* \circ \sigma_{p^*})(s) \le \E_s y_{s,p}(\psi' \circ \tau' \circ \sigma_p)(s).
    \end{equation}
    
    First, we can assume without loss of generality that $\tau^*$ always outputs a solution in $\mathcal{R}$ because $\psi'$ produces $\bot$ (and a utility of $0$) for all proposals in $\Psi \setminus \mathcal{R}$. Then $\psi' \circ \tau^* = \tau^*$ and, by definition of $\mathcal{R}$, we can write $\tau^*  = \psi \circ \hat{\tau}$ for some function $\hat{\tau}$ from strings to $\Sigma$. Then the left hand side of \eqref{eq:single-proposal-ineq} becomes the expected utility of response $(p^*, \hat{\tau})$ against mechanism $M = (\Sigma, \psi)$:
    \begin{equation*}
        \E y_{s,p^*}(\psi' \circ \tau^* \circ \sigma_{p^*})(s) = \E y_{s,p^*}(\psi \circ \hat{\tau} \circ \sigma_{p^*})(s)
    \end{equation*}
    whereas the right hand side of \eqref{eq:single-proposal-ineq} is the expected utility of response $(p, \tau)$ against $M$:
    \begin{equation*}
        \E y_{s,p}(\psi' \circ \tau' \circ \sigma_p)(s) = \E y_{s,p}(\psi \circ \tau \circ \sigma_p)(s).
    \end{equation*}
    Since $(p, \tau)$ is a best response for this mechanism, the desired inequality \eqref{eq:single-proposal-ineq} follows.
\end{proof}

\section{Deterministic Mechanisms}
\label{deterministic-mechanisms}

In this section, we will consider deterministic single-proposal mechanisms for delegated stochastic probing problems, as defined in Section \ref{formal-definition}. This is in contrast to randomized mechanisms which we will define later in Section \ref{lottery-mechanisms}. We will show that large classes of these problems have constant-factor policies, and therefore constant-factor delegation gaps.

The focus of this section is on Theorem \ref{thm:inner-reduction} and Theorem \ref{thm:multiplicative}, which together give us a general method of constructing competitive delegation policies from certain prophet inequalities and adaptivity gaps. In particular, Corollary \ref{cor:inner-intersection-constraints} gives us constant-factor policies for delegated stochastic probing with no outer constraint and an inner constraint which is the intersection of a constant number of matroid, knapsack, and matching constraints. Similarly, Corollary \ref{cor:inner-intersection-outer-constraints} gives us constant-factor policies for delegated stochastic probing with any downwards-closed outer constraint and an inner constraint which is the intersection of a constant number of matroids.

\subsection{Inner Constraint Delegation}
\label{inner-constraint-delegation}

We will now consider instances of delegated stochastic probing for which there is no outer constraint. We will then combine the results from this section with Theorem \ref{thm:multiplicative} to get solutions to delegation problems with both inner and outer constraints.

To simulate the lack of an outer constraint, we will consider instances of delegation for which the outer constraint is the trivial set system in which all subsets of the elements are feasible. For any ground set $E$ of elements, we will write this trivial set system as $\mathcal{M}^*_E$, omitting the subscript when the set of elements $E$ is clear from context.

\begin{theorem}
    \label{thm:inner-reduction}
    Given an instance $I = (E, \mathcal{M}^*, \mathcal{M}_{\text{in}})$ of delegated stochastic probing without outer constraints, let $J$ be an instance of the prophet inequality problem with random variables $X_e$ for all $e \in E$ and constraint $\mathcal{M}_{\text{in}}$. If there exists an $\alpha$-factor greedy strategy for $J$ against the almighty adversary, then there exists a deterministic $\alpha$-policy for $I$. Furthermore, the proof is constructive when given the strategy for $J$.
\end{theorem}
\begin{proof}
    First, we have by our choice of $J$ that the expected utility of the prophet in $J$ is equal to the expected utility of the non-delegating principal in $I$. Notice that the principal has no outer constraint, so we can assume without loss of generality that they probe all elements. Then the prophet and non-delegating principal both get exactly
    \begin{equation*}
        \E \max_{T \in \mathcal{M}_{\text{in}}} x(T).
    \end{equation*}
    
    Now consider the gambler's $\alpha$-factor greedy strategy, which consists of some collection $\mathcal{A} \subseteq 2^{E \times \mathbb{R}_+}$ of ``acceptable'' sets of outcomes. We will define the delegating principal's policy as follows
    \begin{equation*}
        \mathcal{R} = \set{\set{(e, x, y) : (e, x) \in A, y \in \mathbb{R}_+} : A \in \mathcal{A}}.
    \end{equation*}
    Notice that policy $\mathcal{R}$ is exactly the same as strategy $\mathcal{A}$, just translated into the language of delegation.
    
    Now we will show that the utility of the delegating principal with policy $\mathcal{R}$ is at least the utility of the gambler with greedy strategy $\mathcal{A}$. In the prophet inequality, the almighty adversary can order the random variables such that the gambler always gets their least favorite among all maximal acceptable sets (the set is always maximal because the gambler's strategy is greedy). Compare this with delegation, where the agent knows the result of all probed elements as well as the principal's acceptable sets $\mathcal{R}$. Since the agent has non-negative utility for all outcomes, we can assume without loss of generality that they will always propose a maximal acceptable set. For every corresponding set of realizations in each problem, the gambler will receive the maximal set in $\mathcal{A}$ of minimum value and the principal will receive some maximal set in $\mathcal{R}$. Since we defined $\mathcal{R}$ to correspond directly with $\mathcal{A}$, the principal's value must be at least as large as the gambler's. This is true of all possible realizations, so $\mathcal{R}$ must be an $\alpha$-policy for $I$.
\end{proof}

We note that by construction of the principal's policy $\mathcal{R}$, this theorem holds even when the principal is unaware of the agent's utility values $y$. This is comparable to the reduction in \cite{KK18} which similarly worked regardless of the principal's knowledge of the agent's utilities.

Unfortunately, applications of this theorem rely on the existence of competitive strategies against the almighty adversary, which is a very strong condition. It is natural to ask whether it's really necessary in the reduction for the adversary to be almighty. We provide some evidence that this is indeed necessary by considering the special case of a 1-uniform inner matroid. In this case, it's easy to construct instances for which the utility of the principal and agent sum to a constant value for all outcomes, i.e. $X_e + Y_e = c$ for all $e$ and some constant $c$. In such an instance, the agent's goals are directly opposed to the principal's, so the agent will always propose the principal's least favorite acceptable outcome. In the corresponding instance of the prophet inequality, the almighty adversary can guarantee that the gambler chooses their least favorite acceptable outcome, while weaker adversaries (that don't know the realizations of variables) cannot enforce the same guarantee.

Using some known greedy prophet inequalities against the almighty adversary, we get the following corollaries.

\begin{corollary}
    \label{cor:inner-uniform-constraint}
    There exist deterministic $\frac{1}{2}$-policies for delegated stochastic probing problems with no outer constraint and a 1-uniform inner constraint.
\end{corollary}
\begin{proof}
    This follows from the existence of $\frac{1}{2}$ threshold rules (such as Samuel-Cahn's median rule \cite{SC84}) for the 1-uniform prophet inequality against the almighty adversary.
\end{proof}

\begin{corollary}
    \label{cor:inner-constraints}
    There exist constant-factor deterministic policies for delegated stochastic probing problems with no outer constraint and three classes of inner constraints. These factors are: $\frac{1}{4}$ for matroid constraints, $\frac{1}{2e} \approx 0.1839$ for matching constraints, and $\frac{3}{2} - \sqrt{2} \approx 0.0857$ for knapsack constraints.
\end{corollary}
\begin{proof}
    This corollary is largely based on results from \cite{OCRS}. By combining \cite[Theorem 1.8]{OCRS} with \cite[Observation 1.6]{OCRS} and optimizing the parameters, we get randomized greedy online contention resolution schemes (OCRS) for three aforementioned  constraint systems with the same factors listed above. Then, applying \cite[Theorem 1.12]{OCRS}, each randomized greedy OCRS corresponds to a randomized greedy prophet inequality against the almighty adversary with the same approximation factor. Since the adversary is almighty, they can predict any randomness in our strategy. Therefore, the randomized strategy is no better than the best deterministic strategy, and there must exist some deterministic strategy achieving the same factor. Finally, we apply our Theorem \ref{thm:inner-reduction} to turn the prophet inequality strategy into a delegation policy with the same factor.
\end{proof}

\begin{corollary}
    \label{cor:inner-intersection-constraints}
    There exist constant-factor deterministic policies for delegated stochastic probing problems with no outer constraint and an inner constraint that is the intersection of a constant number of matroid, knapsack, and matching constraints.
\end{corollary}
\begin{proof}
    We use \cite[Corollary 1.13]{OCRS} along with the same reasoning as Corollary \ref{cor:inner-constraints}.
\end{proof}

We note that it is open whether there exists a $\frac{1}{2}$-OCRS for matroids against the almighty adversary \cite{OCRS_prophet}. The existence of such an OCRS, if greedy, would imply the existence of $\frac{1}{2}$-policy for delegated stochastic probing with a matroid inner constraint and no outer constraint.

Although Corollary \ref{cor:inner-uniform-constraint} applies to a model very similar to the distributional delegation model from \cite{KK18}, our principal has the additional power of being able to distinguish between otherwise identical elements by their name alone. However, by observing that Theorem~\ref{thm:inner-reduction} turns greedy prophet inequalities that don't distinguish between identical elements into delegation policies that also don't distinguish between identical elements, we can derive delegation policies that recover the $\frac{1}{2}$-factor guarantee from \cite{KK18} for their distributional model. We leave the details for Section \ref{symmetric-delegation-policies}.

\subsection{Outer Constraint Delegation}
\label{outer-constraint-delegation}

Using the adaptivity gap from Section \ref{adaptivity-gap}, we will now show that there are large classes of delegated stochastic probing problems with nontrivial outer constraints for which the principal can achieve, in expectation, a constant-factor of their non-delegated optimal utility.

\begin{theorem}
    \label{thm:multiplicative}
    Let $I = (E, \mathcal{M}_{\text{out}}, \mathcal{M}_{\text{in}})$ be an instance of delegated stochastic probing. Suppose that, for all $F \in \mathcal{I}_{\text{out}}$, there exists a deterministic $\alpha$-policy for the restriction $I_F = (F, \mathcal{M}^*_F, \mathcal{M}_{\text{in}}|F)$ of instance $I$ to $F$. Suppose also that the adaptivity gap for weighted rank functions of $\mathcal{M}_{\text{in}}$ on set system $\mathcal{M}_{\text{out}}$ is at least $\beta$. Then there exists a deterministic $\alpha \beta$-policy for instance $I$.
\end{theorem}
\begin{proof}
    Given any set of elements $F \subseteq E$, we can write the utility of the non-delegating principal who probes $F$ as
    \begin{equation*}
        u(F) = \max_{G \subseteq F, G \in \mathcal{I}_{\text{in}}} x(G)
    \end{equation*}
    and the utility of the delegating principal with policy $\mathcal{R}$ who probes $F$ as
    \begin{equation*}
        v_{\mathcal{R}}(F) = x\paren{\argmax_{G \subseteq F, \Omega_G \in \mathcal{R}} y(G)},
    \end{equation*}
    where $\Omega_G \subseteq \Omega$ is the set of outcomes from the probed elements $G$.
    
    Notice that for any fixed set of realizations from all random variables, $u$ is just the weighted rank function of set system $\mathcal{M}_{\text{in}}$. Therefore, by the adaptivity gap for such a function over set system $\mathcal{M}_{\text{out}}$, there exists a fixed set $F \in \mathcal{I}_{\text{out}}$ such that
    \begin{equation}
        \label{eq:multiplicative-thm-1}
        \E u(F) \ge \beta \E u(E^*),
    \end{equation}
    where $E^* \in \mathcal{I}_{\text{out}}$ is a random variable representing the optimal set of elements selected by an adaptive non-delegating principal. Notice that expectation is also over the randomness of~$E^*$.
    
    Now we will consider the same delegation instance with access to only the elements in $F$, i.e. instance $(F, \mathcal{M}_{\text{out}}| F, \mathcal{M}_{\text{in}}| F)$. Since $F \in \mathcal{I}_{\text{out}}$, the outer matroid doesn't restrict probing at all and this instance is equivalent to $I_F = (F, \mathcal{M}_F^*, \mathcal{M}_{\text{in}}|F)$. By our assumption, this problem has some $\alpha$-approximate delegation policy. Let $\mathcal{R}$ be one such policy. Then we have
    \begin{equation}
        \label{eq:multiplicative-thm-2}
        \E v_{\mathcal{R}}(F) \ge \alpha \E u(F).
    \end{equation}
    
    Since $\mathcal{R}$ contains outcomes only from elements in $F$, an agent restricted to $\mathcal{R}$ in the original instance $I$ has no incentive to probe elements outside of $F$. Because $F \in \mathcal{I}_{\text{out}}$, the agent can probe all of $F$. Therefore, we can assume without loss of generality that an optimal adaptive strategy will choose to probe exactly the elements in $F$. Then
    \begin{equation}
        \label{eq:multiplicative-thm-3}
        \E v_{\mathcal{R}}(E^*_{\mathcal{R}}) = \E v_{\mathcal{R}}(F),
    \end{equation}
    where $E^*_{\mathcal{R}} \subseteq E$ is a random variable containing exactly the elements probed by an optimal adaptive agent when when restricted to acceptable set $\mathcal{R}$ in the original instance $I$.
    
    Combining \eqref{eq:multiplicative-thm-1}, \eqref{eq:multiplicative-thm-2}, and \eqref{eq:multiplicative-thm-3}, we get the desired inequality:
    \begin{align*}
        \E v_{\mathcal{R}}(E^*_{\mathcal{R}})
        &= \E v_{\mathcal{R}}(F) \\
        &\ge \alpha \E u(F) \\
        &\ge \alpha \beta \E u(E^*).
    \end{align*}
\end{proof}

\begin{corollary}
    \label{cor:inner-uniform-outer-matroid-constraint}
    There exist deterministic $\frac{1}{2} \paren{1 - \frac{1}{e}} \approx 0.3160$-policies for delegated stochastic probing problems with matroid outer constraints and a 1-uniform inner constraint.
\end{corollary}
\begin{proof}

  By Corollary \ref{cor:inner-uniform-constraint}, there is a $\frac{1}{2}$-policy for any instance of delegated stochastic probing with a 1-uniform inner constraint and no outer constraint. Every restriction of our present instance $I$ to some independent set $F$ of the outer matroid is of this form.

  From \cite{adaptivity}, we have a $1 - \frac{1}{e}$ adaptivity gap for stochastic submodular on matroid constraints. Since the weighted rank function of any matroid is submodular, the adaptivity gap of weighted rank functions of the inner 1-uniform matroid constraint on the outer matroid constraint is also $1- \frac{1}{e}$.

  Therefore, the conditions of Theorem \ref{thm:multiplicative} hold with $\alpha=\frac{1}{2}$ and $\beta=1-\frac{1}{e}$, and we get the desired factor.
\end{proof}

\begin{corollary}
    \label{cor:inner-matroid-outer-matroid-constraint}
    There exist deterministic $\frac{1}{4} \paren{1 - \frac{1}{e}} \approx 0.1580$-policies for delegated stochastic probing problems with matroid outer and inner constraints.
\end{corollary}
\begin{proof}
    Similar to Corollary \ref{cor:inner-uniform-outer-matroid-constraint}, we use the $1 - \frac{1}{e}$ adaptivity gap for submodular functions over matroid constraints along with Corollary \ref{cor:inner-constraints}.
\end{proof}

\begin{corollary}
    \label{cor:inner-intersection-outer-constraints}
    There exist constant-factor deterministic policies for delegated stochastic probing with any downward-closed outer constraint and an inner constraint which is the intersection of a constant number of matroids.
\end{corollary}
\begin{proof}
    By \cite[Theorem 1.2]{multivalue_adaptivity}, we have constant-factor adaptivity gaps for weighted rank functions of the intersection of a constant number of matroids over ``prefix-closed'' constraints, which include all downward-closed constraints. By Corollary \ref{cor:inner-intersection-constraints}, we have constant-factor policies for delegated stochastic probing with no outer constraint and an inner constraint which is the intersection of a constant number of matroids. Combining these results with Theorem \ref{thm:multiplicative}, we get the desired constant factors.
\end{proof}

\section{Lottery Mechanisms}
\label{lottery-mechanisms}

One natural generalization of the delegated stochastic probing model defined in section \ref{formal-definition} is to allow the principal to use randomized mechanisms. For example, one may consider the generalization of single-proposal mechanisms which attaches a probability $p_R$ to each set of outcomes $R \subseteq \Omega_{\text{in}}$, and accepts a proposed set $R$ with precisely that probability (and accepts the empty set otherwise). More general lotteries (i.e. with non-binary support) are also possible. It's then natural to ask whether there exist instances for which some randomized policy does better than all deterministic ones. Even further, we can ask whether there exists a randomized policy that strictly outperforms deterministic ones in the worst case. In other words, can randomization give us a strictly better delegation gap?

In this section, we will broadly discuss randomized mechanisms and then consider the special case of delegation with 1-uniform inner constraints and no outer constraints. In this special case, there exist instances for which randomization significantly helps the principal, and there are worst-case instances in which the delegation gap of $\frac{1}{2}$ is tight for randomized mechanisms as well as deterministic ones. Before getting to these results, we will discuss methods of randomization and then formalize what we mean by a randomized mechanism.

There are two obvious ways that the single-proposal mechanism can be randomized. The first is for the principal to sample a deterministic policy $\mathcal{R}$ from some distribution and then run the single proposal mechanism defined by $\mathcal{R}$. However, noticing that our model of delegation is a Stackelberg game, we can conclude that there always exists a pure optimal strategy for the principal, so this type of randomization doesn't give the principal any more power.

The second type of randomness is for the policy itself to be a set of acceptable distributions over sets of outcomes (i.e. a menu of lotteries), from which the agent may propose one. The principal then samples a set of outcomes from the proposed lottery. This expands the space of mechanisms, conceivably  affording the principal more power in influencing the agent's behavior. We will focus on these randomized mechanisms for the rest of this section.

\begin{definition}
    A \emph{lottery mechanism} is a randomized mechanism for delegated stochastic probing consisting of a set $\mathcal{R}$ of distributions, or lotteries, each with support in $\Omega_{\text{in}}$. After the set of acceptable lotteries $\mathcal{R}$ is selected and announced to the agent, the delegated stochastic probing mechanism proceeds largely the same. The agent probes outcomes $S$ and proposes one of the lotteries $L \in \mathcal{R}$. The principal then samples a set of outcomes $T \sim L$ from that lottery. If $T \subseteq S$, then the principal and agent receive $x(T)$ and $y(T)$ utility, respectively. Otherwise, they both receive $0$ utility.
\end{definition}

We note that this sort of mechanism is a generalized single-proposal mechanism in the sense of Section \ref{single-proposal-mechanisms}: Each lottery represents a solution and an agent's expected utility for a lottery represents their utility for that solution. Therefore, Lemma \ref{lem:single-proposal-mechanisms} applies to lottery mechanisms as well.

\subsection{Power of Lotteries}
\label{power-of-lotteries}

The increased power of lottery mechanisms means that for some instances of delegated stochastic probing there exist lottery policies that provide the principal with a  better expected utility than the best deterministic policies. In fact, we will show that there are instances for which some lottery policies nearly achieve the principal's non-delegated expected utility, while the best deterministic policies achieve only about half of this value.

First, we will make the observation that it never benefits the principal to declare two lotteries in $\mathcal{R}$ with identical support but different distributions. This is because the principal knows the utility function of the agent and can predict which lottery the agent will prefer. Therefore, we can assume that for any given support, the principal will declare at most one lottery.

\begin{proposition}
    \label{prop:lotteries-positive}
    For all $0 < \epsilon < 1$, there exists an instance of delegated stochastic probing for which the best lottery mechanisms achieve $\frac{2 - 3\epsilon + 2\epsilon^2}{2 - \epsilon}$ of the principal's non-delegated expected utility, while the best deterministic mechanisms achieve $\frac{1}{2 - \epsilon}$ of the principal's non-delegated expected utility. As $\epsilon$ approaches $0$, the former approaches $1$ while the latter approaches $\frac{1}{2}$.
\end{proposition}
\begin{proof}
    Consider an instance with elements $E = \set{1, 2}$, a 1-uniform matroid inner constraint, no outer constraint, and distributions for elements $1$ and $2$ as detailed in Table \ref{tab:lottery-positive}.
    \begin{table}[ht]
        \centering
        \begin{tabular}{c|c c c|c}
            Outcome & Element $e$ & Utility $x$ & Utility $y$ & Probability $\P_e[(x, y)]$ \\
            \hline
            $\omega_0$ & $1$ & $0$ & $0$ & $1 - \epsilon$ \\
            $\omega_1$ & $1$ & $1/\epsilon$ & $1 - \epsilon$ & $\epsilon$ \\
            \hline
            $\omega_2$ & $2$ & $1$ & $1$ & $1$
        \end{tabular}
        \caption{Each row represents a single outcome and contains its name, element $e$, utilities $x$ and $y$, and the probability that it is probed from element $e$.}
        \label{tab:lottery-positive}
    \end{table}
    
    Since there are no outer constraints we assume that both elements are probed. The non-delegating principal can accept $\omega_1$ when it appears and accept $\omega_2$ otherwise. This gives them an expected utility of $\epsilon / \epsilon + 1 - \epsilon = 2 - \epsilon$. By enumerating all deterministic policies, we can confirm that the best such policy gives the delegating principal an expected utility of $1$. Therefore, the best deterministic policy achieves $\frac{1}{2 - \epsilon}$ of the principal's non-delegated utility.
    
    Now consider a lottery policy with lotteries $A$ and $B$ such that $\P_A[\omega_1] = 1$, $\P_B[\omega_2] = 1 - 2 \epsilon$, and $\P_B[\omega_0] = 2 \epsilon$. We can quickly verify that this gives the delegating principal an expected utility of $2 - 3 \epsilon + 2 \epsilon^2$. Therefore, at least one lottery policy achieves $\frac{2 - 3\epsilon + 2\epsilon^2}{2 - \epsilon}$ of the principal's non-delegated utility.
\end{proof}

Unfortunately, there is good reason not to be too optimistic about the increased power of lottery mechanisms. As we will now show, there also exist instances for which the best lottery policies and the best deterministic policies all achieve approximately half of the principal's non-delegated expected utility. Since Corollary \ref{cor:inner-uniform-constraint} gives us a deterministic $\frac{1}{2}$-policy, this tells us that, in the worst case, the factor $\frac{1}{2}$ is tight even for lottery policies in the special case of no outer constraint and a 1-uniform inner constraint.

\begin{proposition}
    \label{prop:lotteries-negative}
    For all $0 < \epsilon < 1$, there exists an instance of delegated stochastic probing with a 1-uniform inner constraint and no outer constraint for which the best lottery policies and the best deterministic policies all achieve $\frac{1}{2 - \epsilon}$ of the principal's non-delegated expected utility. As $\epsilon$ approaches $0$, this approaches $\frac{1}{2}$.
\end{proposition}
\begin{proof}
    Consider an instance with elements $E = \set{1, 2}$, a 1-uniform matroid inner constraint, no outer constraint, and distributions for elements $1$ and $2$ as detailed in Table \ref{tab:lottery-negative}.
    \begin{table}[ht]
        \centering
        \begin{tabular}{c|c c c|c}
            Outcome & Element $e$ & Utility $x$ & Utility $y$ & Probability $\P_e[(x, y)]$ \\
            \hline
            $\omega_0$ & $1$ & $0$ & $0$ & $1 - \epsilon$ \\
            $\omega_1$ & $1$ & $1/\epsilon$ & $0$ & $\epsilon$ \\
            \hline
            $\omega_2$ & $2$ & $1$ & $1$ & $1$
        \end{tabular}
        \caption{Each row represents a single outcome and contains its name, element $e$, utilities $x$ and $y$, and the probability that it is probed from element $e$. Notice that this instance is identical to the one from Table \ref{tab:lottery-positive} except for the agent's utility for outcome $\omega_1$.}
        \label{tab:lottery-negative}
    \end{table}

    In the case of ties, we assume that the agent prefers to break ties first in the principal's favor and then arbitrarily among any remaining ties. This assumption serves only to simplify the proof, and can be avoided with careful modifications to the utility table.
    
    The non-delegating principal can accept $\omega_1 = (1, 1/\epsilon, 0)$ when it appears and accept $\omega_2 = (2, 1, 1)$ otherwise. This gives them an expected utility of $\epsilon / \epsilon + 1 - \epsilon = 2 - \epsilon$. By enumerating all deterministic policies, we can confirm that the best such policy gives the delegating principal an expected utility of $1$. Therefore, the best deterministic policy achieves $\frac{1}{2 - \epsilon}$ of the principal's non-delegated utility.
    
    Finding the best menu of lotteries takes slightly more work. Since the inner constraint is 1-uniform, each lottery is supported on singletons as well as  the empty set. Recall also that we can restrict attention to menus where no two lotteries have the same support. We claim that we can further restrict attention to menus with exactly two lotteries $A$ and $B$, with $A$ supported on $\set{\omega_0, \omega_2}$ and $B$ supported on $\set{\omega_1, \omega_2}$. To see this, observe that:
    \begin{enumerate}
        \item Shifting all probability mass from the empty set to $\omega_0$ or $\omega_1$ in any lottery does not affect the agent's utility and can only increase the principal's utility. In the case of tie-breaking, the principal's favorite remaining lottery is no worse than before.
        
        \item If there is a lottery with both $\omega_0$ and $\omega_1$ in its support, shifting all probability mass from one of these outcomes to the other does not affect the agent's utility, and in at least one direction this shift of probability mass will make the policy no worse for the principal. Again, in the case of tie-breaking, the principal's favorite remaining lottery is no worse than before.
        
        \item A menu without lottery $A$ is no better than the same menu with lottery $A$ for which all probability mass of $A$ is assigned to $\omega_0$. Similarly, a menu without lottery $B$ is no better than the same menu with lottery $B$ for which all probability mass of $B$ is assigned to $\omega_1$.
    \end{enumerate}
    
    Parametrizing the probability of each outcome, we get: $\P_A[\omega_2] = a$, $\P_A[\omega_0] = 1 - a$, $\P_B[\omega_2] = b$, and $\P_B[\omega_1] = 1 - b$ for some $a, b \in [0, 1]$. No matter what the agent probes ($\set{\omega_0, \omega_2}$ or $\set{\omega_1, \omega_2}$), their favorite lottery is $B$ if $b \ge a$ and $A$ otherwise. If we choose $b \ge a$, the delegating principal gets expected utility $\epsilon (b + (1 - b) / \epsilon) + (1 - \epsilon) \cdot b = 1$. Otherwise, the principal gets $\epsilon \cdot a + (1 - \epsilon) \cdot a = a$, which can be made as large as $1$ for $a = 1$. Therefore, the best lottery policy achieves $\frac{1}{2 - \epsilon}$ of the principal's non-delegated utility.
\end{proof}

\section{Open Questions}
\label{open-questions}

Due to this novel combination of delegation with stochastic probing, we believe that this paper ultimately opens up many more questions than it answers. In this section, we will make some of these questions explicit.
\begin{itemize}
    \item While we focused on the existence of constant-factor delegation policies regardless of their computational complexity, applying these solutions to practical problems requires some guarantee that they can be easily computed and represented. Are there delegated stochastic probing problems for which constant-factor policies are NP-hard to compute in general? Are there special cases for which constant-factor policies can always be computed in polynomial time?
    
    \item In Section \ref{lottery-mechanisms}, we showed that the constant given in Corollary \ref{cor:inner-uniform-constraint} is tight. Are the other factors given in Section \ref{inner-constraint-delegation} tight? We note that this is related to an open question by \cite{OCRS_prophet} about $\frac{1}{2}$ prophet inequalities on matroids against the almighty adversary.
    
    \item Are the constant factors given in Section \ref{outer-constraint-delegation} tight? Due to the broad applicability of adaptivity gaps, our method is unlikely to take advantage of special structure that may be present in delegated stochastic probing problems. Therefore, it seems probable that better constants exist, but we make no claim to a conjecture.
    
    \item Our model assumes that probing is always zero-cost, so it doesn't generalize the binary model from \cite{KK18} or the box problem of \cite{weitzman79}. It's natural to ask whether we can get constant-factor delegation gaps with probing costs in addition to (or as a replacement for) outer constraints.
    
    \item Our model doesn't allow the principal to incentivize the agent with transfers (such as payments), so it's natural to ask how such an extension to the model could improve the principal's worst-case guarantees.
    
    \item If the principal is delegating to multiple agents simultaneously, can they get better worst-case guarantees than delegating to a single agent? We note that there are many ways to define this formally. For example, a stronger principal may be able to define different acceptable sets for each agent whereas a weaker principal may be forced to declare one acceptable set for all agents.
    
    \item It's not hard to imagine practical applications of stochastic probing for which elements are not independently distributed. Can we get competitive guarantees even in the absence of the independence assumption?
\end{itemize}

\bibliographystyle{abbrvnat}
\bibliography{stochastic}

\begin{thebibliography}{27}
\providecommand{\natexlab}[1]{#1}
\providecommand{\url}[1]{\texttt{#1}}
\expandafter\ifx\csname urlstyle\endcsname\relax
  \providecommand{\doi}[1]{doi: #1}\else
  \providecommand{\doi}{doi: \begingroup \urlstyle{rm}\Url}\fi

\bibitem[Adamczyk and W{\l}odarczyk(2018)]{ROCRS}
M.~Adamczyk and M.~W{\l}odarczyk.
\newblock Random order contention resolution schemes.
\newblock In \emph{2018 IEEE 59th Annual Symposium on Foundations of Computer
  Science (FOCS)}, pages 790--801. IEEE, 2018.

\bibitem[Alonso and Matouschek(2008)]{alonso08}
R.~Alonso and N.~Matouschek.
\newblock Optimal delegation.
\newblock \emph{The Review of Economic Studies}, 75\penalty0 (1):\penalty0
  259--293, 2008.

\bibitem[Armstrong and Vickers(2010)]{armstrong10}
M.~Armstrong and J.~Vickers.
\newblock A model of delegated project choice.
\newblock \emph{Econometrica}, 78\penalty0 (1):\penalty0 213--244, 2010.

\bibitem[Asadpour and Nazerzadeh(2016)]{adaptivity}
A.~Asadpour and H.~Nazerzadeh.
\newblock Maximizing stochastic monotone submodular functions.
\newblock \emph{Management Science}, 62\penalty0 (8):\penalty0 2374--2391,
  2016.

\bibitem[Bradac et~al.(2019)Bradac, Singla, and Zuzic]{multivalue_adaptivity}
D.~Bradac, S.~Singla, and G.~Zuzic.
\newblock (near) optimal adaptivity gaps for stochastic multi-value probing.
\newblock \emph{arXiv preprint arXiv:1902.01461}, 2019.

\bibitem[Chekuri et~al.(2014)Chekuri, Vondr{\'a}k, and Zenklusen]{CRS}
C.~Chekuri, J.~Vondr{\'a}k, and R.~Zenklusen.
\newblock Submodular function maximization via the multilinear relaxation and
  contention resolution schemes.
\newblock \emph{SIAM Journal on Computing}, 43\penalty0 (6):\penalty0
  1831--1879, 2014.

\bibitem[Chen et~al.(2009)Chen, Immorlica, Karlin, Mahdian, and
  Rudra]{matching_heaven}
N.~Chen, N.~Immorlica, A.~R. Karlin, M.~Mahdian, and A.~Rudra.
\newblock Approximating matches made in heaven.
\newblock In \emph{International Colloquium on Automata, Languages, and
  Programming}, pages 266--278. Springer, 2009.

\bibitem[D{\"u}tting et~al.(2017)D{\"u}tting, Feldman, Kesselheim, and
  Lucier]{prophet_easy}
P.~D{\"u}tting, M.~Feldman, T.~Kesselheim, and B.~Lucier.
\newblock Prophet inequalities made easy: Stochastic optimization by pricing
  non-stochastic inputs.
\newblock In \emph{2017 IEEE 58th Annual Symposium on Foundations of Computer
  Science (FOCS)}, pages 540--551. IEEE, 2017.

\bibitem[Feldman et~al.(2016)Feldman, Svensson, and Zenklusen]{OCRS}
M.~Feldman, O.~Svensson, and R.~Zenklusen.
\newblock Online contention resolution schemes.
\newblock In \emph{Proceedings of the twenty-seventh annual ACM-SIAM symposium
  on Discrete algorithms}, pages 1014--1033. Society for Industrial and Applied
  Mathematics, 2016.

\bibitem[Gupta and Nagarajan(2013)]{probing}
A.~Gupta and V.~Nagarajan.
\newblock A stochastic probing problem with applications.
\newblock In \emph{International Conference on Integer Programming and
  Combinatorial Optimization}, pages 205--216. Springer, 2013.

\bibitem[Gupta et~al.(2016)Gupta, Nagarajan, and Singla]{probing_gns16}
A.~Gupta, V.~Nagarajan, and S.~Singla.
\newblock Algorithms and adaptivity gaps for stochastic probing.
\newblock In \emph{Proceedings of the twenty-seventh annual ACM-SIAM symposium
  on Discrete algorithms}, pages 1731--1747. SIAM, 2016.

\bibitem[Hajiaghayi et~al.(2007)Hajiaghayi, Kleinberg, and Sandholm]{HKS07}
M.~T. Hajiaghayi, R.~Kleinberg, and T.~Sandholm.
\newblock Automated online mechanism design and prophet inequalities.
\newblock In \emph{AAAI}, volume~7, pages 58--65, 2007.

\bibitem[Holmstrom(1980)]{holmstrom80}
B.~Holmstrom.
\newblock On the theory of delegation.
\newblock Technical report, Discussion Paper, 1980.

\bibitem[Holmstrom(1978)]{holmstrom78}
B.~R. Holmstrom.
\newblock \emph{On Incentives and Control in Organizations.}
\newblock PhD thesis, Stanford University, 1978.

\bibitem[Khodabakhsh et~al.(2020)Khodabakhsh, Liu, Pountourakis, Taggart, and
  Zhang]{threshold_delegation}
A.~Khodabakhsh, Y.~Liu, E.~Pountourakis, S.~Taggart, and Y.~Zhang.
\newblock Threshold policies for delegation.
\newblock \emph{working paper}, 2020.

\bibitem[Kleinberg and Kleinberg(2018)]{KK18}
J.~Kleinberg and R.~Kleinberg.
\newblock Delegated search approximates efficient search.
\newblock In \emph{Proceedings of the 2018 ACM Conference on Economics and
  Computation}, pages 287--302, 2018.

\bibitem[Kleinberg and Weinberg(2012)]{matroid_prophet}
R.~Kleinberg and S.~M. Weinberg.
\newblock Matroid prophet inequalities.
\newblock In \emph{Proceedings of the forty-fourth annual ACM Symposium on
  Theory of Computing}, pages 123--136. ACM, 2012.

\bibitem[Kov{\'a}{\v{c}} and Mylovanov(2009)]{lottery_delegation}
E.~Kov{\'a}{\v{c}} and T.~Mylovanov.
\newblock Stochastic mechanisms in settings without monetary transfers: The
  regular case.
\newblock \emph{Journal of Economic Theory}, 144\penalty0 (4):\penalty0
  1373--1395, 2009.

\bibitem[Krengel and Sucheston(1977)]{KS77}
U.~Krengel and L.~Sucheston.
\newblock Semiamarts and finite values.
\newblock \emph{Bulletin of the American Mathematical Society}, 83\penalty0
  (4):\penalty0 745--747, 1977.

\bibitem[Krengel and Sucheston(1978)]{KS78}
U.~Krengel and L.~Sucheston.
\newblock On semiamarts, amarts, and processes with finite value.
\newblock \emph{Probability on Banach spaces}, 4:\penalty0 197--266, 1978.

\bibitem[Lee and Singla(2018)]{OCRS_prophet}
E.~Lee and S.~Singla.
\newblock Optimal online contention resolution schemes via ex-ante prophet
  inequalities.
\newblock In \emph{26th Annual European Symposium on Algorithms (ESA 2018)}.
  Schloss Dagstuhl-Leibniz-Zentrum fuer Informatik, 2018.

\bibitem[Lucier(2017)]{prophet_survey_lucier}
B.~Lucier.
\newblock An economic view of prophet inequalities.
\newblock \emph{ACM SIGecom Exchanges}, 16\penalty0 (1):\penalty0 24--47, 2017.

\bibitem[Melumad and Shibano(1991)]{melumad91}
N.~D. Melumad and T.~Shibano.
\newblock Communication in settings with no transfers.
\newblock \emph{The RAND Journal of Economics}, pages 173--198, 1991.

\bibitem[Roughgarden(2016)]{roughgarden_twenty_lectures}
T.~Roughgarden.
\newblock \emph{Twenty lectures on algorithmic game theory}.
\newblock Cambridge University Press, 2016.

\bibitem[Samuel-Cahn et~al.(1984)]{SC84}
E.~Samuel-Cahn et~al.
\newblock Comparison of threshold stop rules and maximum for independent
  nonnegative random variables.
\newblock \emph{the Annals of Probability}, 12\penalty0 (4):\penalty0
  1213--1216, 1984.

\bibitem[Singla(2018)]{singla_thesis}
S.~Singla.
\newblock \emph{Combinatorial Optimization Under Uncertainty: Probing and
  Stopping-Time Algorithms}.
\newblock PhD thesis, PhD thesis, Carnegie Mellon University, 2018.

\bibitem[Weitzman(1979)]{weitzman79}
M.~L. Weitzman.
\newblock Optimal search for the best alternative.
\newblock \emph{Econometrica: Journal of the Econometric Society}, pages
  641--654, 1979.

\end{thebibliography}

\appendix

\section{Appendix}
\label{appendix}

\subsection{Symmetric Delegation Policies}
\label{symmetric-delegation-policies}

While our model is not a direct generalization of the distributional model used by Kleinberg and Kleinberg, we can obtain a generalization by considering delegated stochastic probing with a restricted class of policies, which we call symmetric policies. Given this variant, we can recover the $\frac{1}{2}$ factor that they obtained. First, we need to define some notation and terminology.

Given any object $X$ (such as a set, tuple, or recursive combination of the two) containing atomic elements $E$, we can consider the operation of taking two elements $e_1, e_2 \in E$ and swapping all instances of $e_1$ and $e_2$ in $X$. More generally, for any permutation $\pi$ of elements in $E$, we can consider rewriting all elements $e$ to $\pi(e)$ simultaneously. We will denote the object obtained from this operation as $X[E \to \pi(E)]$.

\begin{definition}
    Fix an instance of delegated stochastic probing with elements $E$, outer constraint $\mathcal{M}_{\text{out}}$, and inner constraint $\mathcal{M}_{\text{in}}$. We say that a subset of elements $F \subseteq E$ are \emph{symmetric} if $\mu_e = \mu_f$ for all $e, f \in F$ and for all permutations $\pi$ on $F$ we have that $\mathcal{M}_{\text{in}}[F \to \pi(F)] = \mathcal{M}_{\text{in}}$ and $\mathcal{M}_{\text{out}}[F \to \pi(F)] = \mathcal{M}_{\text{out}}$.
\end{definition}

\begin{definition}
    Fix an instance of delegated stochastic probing with elements $E$, outer constraint $\mathcal{M}_{\text{out}}$, and inner constraint $\mathcal{M}_{\text{in}}$. We say that a policy $\mathcal{R}$ is \emph{symmetric} if $\mathcal{R}[F \to \pi(F)] = \mathcal{R}$ for all symmetric sets of elements $F \subseteq E$ and all permutations $\pi$ on $F$.
\end{definition}

Intuitively, symmetric elements are ones which are identical in everything except name. Then symmetric policies are ones that don't distinguish between symmetric elements. Using this intuition, we will now consider the problem of delegated stochastic probing with $k$ identically distributed elements $E$, a 1-uniform inner constraint, and no outer constraint. Given any such instance, it's easy to see that all elements $E$ are symmetric. Notice the similarity between such an instance and the distributional model. The only difference is that our principal has the power to distinguish between outcomes sampled from different elements. However, if the principal is restricted to symmetric policies, then their policy cannot distinguish between different elements, so it must characterize acceptable outcomes based only on their $(x, y)$ utility. This is equivalent to the distributional model.

There are also natural definitions of symmetric elements and strategies in the prophet inequality problem.

\begin{definition}
    Fix an instance of the prophet inequality problem with elements $E$ and feasibility constraint $\mathcal{M}$. We say that a subset of elements $F \subseteq E$ are \emph{symmetric} if $X_e$ and $X_f$ are identically distributed for all $e, f \in F$ and for all permutations $\pi$ on $F$ we have that $\mathcal{M}[F \to \pi(F)] = \mathcal{M}$.
\end{definition}

\begin{definition}
    Fix an instance of the prophet inequality problem with elements $E$ and feasibility constraint $\mathcal{M}$. We say that a strategy $\mathcal{A}$ is \emph{symmetric} if $\mathcal{A}[F \to \pi(F)] = \mathcal{A}$ for all symmetric sets of elements $F \subseteq E$ and all permutations $\pi$ on $F$.
\end{definition}

Given these definitions, we will show that Theorem \ref{thm:inner-reduction} actually transforms symmetric greedy prophet inequalities against the almighty adversary into symmetric delegation policies. This is stated formally in Proposition \ref{prop:symmetric-reduction}.

\begin{proposition}
    \label{prop:symmetric-reduction}
    Given an instance $I = (E, \mathcal{M}^*, \mathcal{M}_{\text{in}})$ of delegated stochastic probing without outer constraints, let $J$ be an instance of the prophet inequality problem with random variables $X_e$ for all $e \in E$ and constraint $\mathcal{M}_{\text{in}}$. If there exists a symmetric $\alpha$-factor greedy strategy for $J$ against the almighty adversary, then there exists a symmetric deterministic $\alpha$-policy for $I$. Furthermore, the proof is constructive when given the strategy for $J$.
\end{proposition}
\begin{proof}
    The proof is identical to the proof of Theorem \ref{thm:inner-reduction}, but we observe that the greedy strategy $\mathcal{A}$ for prophet inequality problem $J$ is symmetric, so the policy $\mathcal{R}$ derived from $\mathcal{A}$ must also be symmetric by construction.
\end{proof}

Since the $\frac{1}{2}$ prophet inequality used in Corollary \ref{cor:inner-uniform-constraint} is a threshold policy, it must be symmetric. Therefore, we have a symmetric $\frac{1}{2}$-policy for delegated stochastic probing problems with no outer constraint and a 1-uniform inner constraint. This recovers a $\frac{1}{2}$-factor for the distributional model of \cite{KK18}, as well as for the slight generalization of this model with multiple distributions and a separate cardinality constraint for each one.

\end{document}